\documentclass[12pt, a4paper]{article}

\usepackage{amsmath}
\usepackage{amsfonts}
\usepackage{amssymb}
\usepackage{color}
\usepackage{graphicx}
\usepackage[hmargin=1in,vmargin={1in,1.1in},a4paper]{geometry}
\usepackage{amsthm}
\usepackage{soul}
\usepackage[normalem]{ulem}
\usepackage{setspace}
\usepackage{pdfpages}
\usepackage{here}
\usepackage{float}
\usepackage{url}

\def\RS#1{\marginpar{\tiny [RS]: #1}}

\newtheorem{proposition}{Proposition}
\newtheorem{corollary}{Corollary}

\theoremstyle{definition}

\onehalfspace

\newcommand{\mycomment}[1]{}

\begin{document}


\title{Testing the free-rider hypothesis in climate policy\thanks{Acknowledgements: M.D.'s research is conducted as part of the Cluster of Excellence `CLICCS - Climate, Climatic Change, and Society', contribution to the CEN of the University of Hamburg. F.N.'s research has been supported by NATCOOP (European Research Council 678049).}
 \\
\bigskip
}

\author{
{\sc Robert C. Schmidt\thanks{Institute for Microeconomics, University of Hagen, Universit{\"a}tsstr.\ 11, 58097 Hagen, Germany; E-mail:
\texttt{robert.schmidt@fernuni-hagen.de}}}
\and
{\sc Moritz Drupp\thanks{Department of Economics, University of Hamburg, Von-Melle-Park 5, 20146 Hamburg, Germany; E-mail:
\texttt{Moritz.Drupp@uni-hamburg.de}}}
\and
{\sc Frikk Nesje\thanks{Department of Economics, University of Copenhagen, {\O}ster Farimagsgade 5, Building 26, 1353 Copenhagen K, Denmark; E-mail:
\texttt{frikk.nesje@econ.ku.dk}}}
\and
{\sc Hendrik Hoegen\thanks{Institute for Microeconomics, University of Hagen, Universit{\"a}tsstr.\ 11, 58097 Hagen, Germany; E-mail:
\texttt{hendrik.hoegen@fernuni-hagen.de}}}
}

\date{\today \\
}

\maketitle

\begin{abstract}
Free-riding is widely perceived as a key obstacle for effective climate policy. In the game-theoretic literature on non-cooperative climate policy and on climate cooperation, the free-rider hypothesis is ubiquitous. Yet, the free-rider hypothesis has not been tested empirically in the climate policy context. With the help of a theoretical model, we demonstrate that if free-riding were the main driver of lax climate policies around the globe, then there should be a pronounced country-size effect: Countries with a larger share of the world's population should, all else equal, internalize more climate damages and thus set higher carbon prices. We use this theoretical prediction for testing the free-rider hypothesis empirically. Drawing on data on emission-weighted carbon prices from 2020, while controlling for a host of other potential explanatory variables of carbon pricing, we find that the free-rider hypothesis cannot be supported empirically, based on the criterion that we propose. Hence, other issues may be more important for explaining climate policy stringency or the lack thereof in many countries.
\newline\newline
\textit{JEL Codes}: F53, H41, H87, Q54, Q58\\
\textit{Keywords}: Carbon pricing, carbon tax, climate policy, climate cooperation, free-riding.
\end{abstract}

\newpage

\clearpage

\section{Introduction}

In the economic literature on climate policy, the so-called free-rider effect is ubiquitous.\footnote{For an overview, see e.g.\  Barrett (2005), Finus (2008), Kolstad and Toman (2005).} Underlying this effect is a simple reasoning: if a country unilaterally raises its efforts to contribute to a global public good (such as climate stabilization), then all countries benefit (e.g., via reduced climate damages). Yet, the country bears the costs individually. Hence, unless the country is large enough so that a substantial fraction of the global benefits accrue inside its borders, then in the absence of a cooperative agreement, the country has little incentives to contribute. Instead, it would prefer to free-ride on any efforts of the other countries. The free-rider problem has also been emphasized in the context of climate cooperation. An influential strand of literature highlights that, due to the free-rider incentive, the ``stable'' coalition size is often small, and especially so when the potential gains from cooperation are large.\footnote{See Barrett (1994), Carraro and Siniscalco (1993), Dixit and Olson (2000), Hoel (1992), Karp and Simon (2013), among others.} Abstaining from a climate coalition that would otherwise induce the country to contribute more is just another way of free-riding on the efforts of others.

Lax climate policies in many countries have often been seen as indicative of the free-rider effect (Nordhaus 2015), and countries' efforts to cooperate on climate policy have indeed proven elusive for many years, with the Copenhagen Climate Summit in 2009 being widely perceived as a failure as just one example. Although these observations seem broadly in line with the free-rider hypothesis, there may be other, equally plausible explanations for the observation that, in practice, climate policy often falls short of what scholars from various disciplines would consider an ``adequate'' response to the climate problem (see, e.g., Drupp et al. 2022; H\"ansel et al. 2020; Nordhaus 2019; Pindyck 2019; Rennert et al. 2022).

In this paper, we ask the following question: Is it possible to identify a ``footprint'' of the free-rider effect that would allow to identify this effect empirically, thereby distinguishing it from other possible explanations for why climate policy may fail? In other words, are lax climate policies in countries around the globe already sufficient ``proof'' that countries are free-riding, or are other explanations for lax climate policies more plausible? It turns out, that according to the canonical way of modeling countries' non-cooperative climate policy, the free-rider effect should indeed produce a unique footprint, beyond the basic finding that climate policies are widely perceived as being ``too lax''. Namely, if the free-rider hypothesis is valid, then free-riding should manifest itself in a pronounced {\em country-size effect}: All else equal, countries with larger populations should implement more stringent emissions control than smaller countries. This is because, according to that model, countries internalize climate damages that accrue within their own boundaries, but neglect damages that occur in other countries. Hence, under non-cooperative decision-making, countries with a larger population internalize more of the externalities than smaller countries.

Intuitively, if a country like the US reduces its emissions substantially, then some of the most adverse effects of climate change, some of which also accrue inside of the US, might be mitigated. Hence, a large country like the US should have {\em some} incentive to regulate its emissions unilaterally. By contrast, a country with a smaller population, such as Sweden, should have very little incentive to act unilaterally: While the per capita costs of emissions control for a given reduction target (relative to the country's emissions) may be comparable to those in the US, the impact on climate damages of a country like Sweden should be negligible, while they can be significant for the US.

The sharp theoretical prediction of a pronounced country-size effect on climate policy offers an opportunity to {\em test} the free-rider hypothesis. In this paper, we use data on countries' emission-weighted carbon prices to analyze if a country-size effect as predicted by the theory under the free-rider hypothesis can be empirically confirmed, thereby controlling for other cross-country differences. If a positive country-size effect on countries' carbon prices can be identified, this points to free-riding as a likely explanation for countries' overall reluctance to implement stringent climate policies. By contrast, if a country-size effect {\em cannot} be identified in the data, then this points to {\em other} possible explanations for why countries find it difficult to tackle the climate problem.

In Section~\ref{sec:model}, we present a simple theoretical framework to analyze how free-riding would manifest itself in countries' carbon pricing policies, based on the canonical way of modeling countries' decision-making in climate policy. We demonstrate that when countries set their carbon prices non-cooperatively, then a pronounced country-size effect is predicted: a country's carbon price should be proportional to its population size. By contrast, if the decision-maker in a country accounts for global social welfare when setting climate policy, then no country size effect on the country's carbon price is predicted. We show that this holds independently of whether the policy makers in the other countries also focus on global welfare, or on the welfare of their own population. Furthermore, we extend the modeling framework to allow for the formation of a climate coalition, in the tradition of Barrett's (1994) seminal contribution, thereby accounting for different population sizes of the countries. We demonstrate that our theoretical predictions continue to hold and are even exacerbated: the {\em largest} countries tend to form a climate coalition and, thus, set higher carbon prices than all other countries.

In Section~\ref{sec:empirical}, we analyze recent data on countries' emission-weighted carbon prices to test the free-rider hypothesis empirically, based on the country-size effect that the theory predicts. We first investigate using linear regression analysis if a country's share of the world population is positively associated with its carbon price, using  emission-weighted average of sector(-fuel) price data from Dolphin (2022). We find, that a positive country-size effect on countries' carbon prices cannot be identified. This holds true also when including a host of other explanatory variables to control for determinants of existing carbon price values, such as GDP per capita, or projected climate damages per capita, among others. 
\mycomment{
[[We also consider logistic regression models to analyze if a country's population share has predictive power on whether the country has implemented a carbon price, again controlling for other country characteristics. While a positive country size effect is identified in some of the regressions, the results are strongly dependent on a single observation (on China) and are, thus, not robust. Without this observation, again no country size effect can be identified.
}

We, thus, conclude that according to our approach based on a country-size effect that the theory predicts, the free-rider hypothesis cannot be supported using actual data on carbon pricing across the globe. This points at other plausible explanations that may be hampering countries' efforts to tackle the climate problem in practice. It is indeed an important insight, because depending on the source of the lack of countries' ambition to protect the climate, different policy approaches may be suitable. 
Among the other plausible explanations for why countries find it difficult to address the climate problem are: competitiveness concerns, carbon leakage, vested interests, lobbyism, distributional concerns, and a lack of public support for ambitious climate policy measures (or a lack of knowledge about climate change). In other words, even if the government of a country does {\em not} seek to free-ride on the efforts of other countries to protect the climate and takes global social welfare considerations into account in its domestic policy-making, it may nevertheless fail to implement stringent climate policy measures for the above-mentioned reasons. 
If that holds true, then the adequate approaches to addressing the climate problem may be different from a world where free-riding is (considered to be) the key issue. For example, suppose countries are reluctant to implement sufficiently high carbon prices, primarily because they fear a reduction in the international competitiveness of certain industries and job losses. Then an adequate response may be the introduction of border carbon adjustment measures (BCA) to alleviate competitiveness effects of their climate policies. By contrast, if free-riding were the main issue, then measures targeted specifically at the free-rider problem, such as trade sanctions or other measures to induce countries to cooperate on climate protection (see Nordhaus, 2015), might be most effective.

\subsection{Related literature}

A country-size effect on carbon prices under non-cooperative decision-making of countries has been identified by Nordhaus (2015). However, the effect has been identified within a more narrow modeling framework (with constant marginal benefits of emission reductions and quadratic costs).\footnote{Nordhaus (2015) also uses a different notion of `country size'. He focuses on GDP, whereas our focus is on population size. However, if countries differ only in population size and are otherwise identical, then the two notions are equivalent.} We re-establish the effect within a more general modeling framework that builds on canonical assumptions from the game-theoretic literature on climate policy, and that does not impose specific functional forms for countries' benefits and costs of abatement. We thereby show that the country-size effect holds more generally. Furthermore, while Nordhaus (2015) mentions the effect, he does not draw the conclusion that this effect may be used to test the free-rider hypothesis. To the best of our knowledge, this paper is the first that attempts to isolate the effect of free-riding on countries' carbon pricing policies in actual data. Based on the country-size effect that the theory predicts, we are thus able to test the free-rider hypothesis, thereby distinguishing it from other likely explanations for why countries seem reluctant to implement stringent climate policies in practice.

The effects of asymmetries between countries have also been analyzed formally in the literature on climate coalition formation (e.g., McGinty 2007; Kolstad 2010; Fuentes-Albero and Rubio 2010).  Finus and McGinty (2019) allow for asymmetric benefits and costs of climate mitigation and show that under certain conditions, in particular a negative covariance between benefit and cost parameters, asymmetries can lead to the formation of larger climate coalitions. Yet, to the best of our knowledge, the case where countries differ only in their population size and are otherwise symmetric, has not been analyzed before. It turns out that the analysis of this case is simple, but yields sharp predictions. The key results from the non-cooperative analysis are exacerbated. Namely, we find that also when countries can form a climate coalition, the theory predicts a pronounced country-size effect on countries' carbon prices. Our theoretical results contribute to the game-theoretic literature on climate cooperation. We demonstrate that when countries differ only in population size, then at most two countries form a climate coalition. These are typically the largest two countries (the largest country is always a member of the coalition).

The country-size effects on carbon prices under non-cooperative and cooperative decision-making resemble theoretical findings from the literature considering country size and the voluntary provision of global public goods (Boadway and Hayashi, 1999; Shrestha and Feehan, 2002/2003; Loeper, 2017; Buchholz and Sandler, 2021 and references therein). Yet, in the context of climate policy,  there is reason to doubt whether free-riding is the main driver of lax policies around the globe. For example, Kornek et al.\ (2020) survey experts from the United Nations Framework Convention on Climate Change (UNFCCC) and from the Intergovernmental Panel on Climate Change (IPCC) on obstacles and response options for climate policy. They find that ``Opposition from special interest groups (for example emission-intensive industries)'' is on aggregate perceived as the most important obstacle. The option ``Global public-good nature of mitigation and free-riding incentives'', by contrast, is perceived as an ``extremely'' or ``very important'' obstacle to climate change mitigation policy less frequently than most other options in their survey.

Our results in a companion paper (Drupp, Nesje, Schmidt, 2022) point in a similar direction. There, we analyze data from a survey among experts who have published on carbon pricing. One of the key findings is that experts' unilateral carbon price recommendations (when border carbon adjustment is available) in aggregate are {\em higher} than experts' global carbon price recommendations. This finding is in stark conflict with the free-rider hypothesis, according to which unilateral carbon prices should be lower than a uniform global price, under which (theoretically) the global welfare optimum can be reached. Hence, the results from the expert survey do not support the free-rider hypothesis. 
Whereas the survey focused on experts' price recommendations, in the present paper, we instead focus on countries' {\em implemented} carbon prices. Yet, our results point in the same direction: we do not find support for the free-rider hypothesis based on countries' weighted nationally implemented carbon prices.

Furthermore, the theoretical analysis in this paper offers a possible explanation for why experts' unilateral carbon price recommendations in Drupp, Nesje, Schmidt (2022) even {\em exceed} their global ones. To this end, in the present paper, we consider (among other things) a scenario where the decision-maker in one country has global social welfare in mind when setting their carbon price, whereas the decision-makers in all other countries seek to maximize their national welfare. We demonstrate that in such a case, the ``altruistic'' decision-maker sets a unilateral carbon price that exceeds the uniform one in the global social welfare optimum. This holds, if climate damages are convex in global emissions, so that the altruistic policy-maker seeks to avert the worst climate damages globally with the help of a unilateral carbon price.\footnote{Hence, the counter-intuitive finding in Drupp, Nesje, Schmidt (2022) may be (partially) explained by asserting that many experts who participated in the survey had global social welfare in mind when providing unilateral carbon price recommendations for their government, while expecting non-cooperative policy-making in other countries.}

Our paper also contributes to a growing strand of literature that analyzes countries' climate policies empirically. Levi et al.\ (2020), for instance, identify political economy determinants of carbon pricing. These authors use a Tobit regression model to analyze carbon prices across more than 200 national and subnational jurisdictions. They identify good governance indicators and public belief in climate change as the most important predictors for carbon pricing, and fossil fuel consumption as a key barrier to the implementation of carbon prices.\footnote{Best and Zhang (2020) perform a related analysis of the magnitude of carbon prices across countries, using a variety of environmental, social, political, and economic variables. Broadly speaking, their results point in a similar direction as the results of Levi et al.\ (2020). See also Dolphin et al.\ (2019) for another related study on the political economy of the introduction of carbon prices, as well as their magnitude. The diffusion of carbon pricing policies across countries is analyzed by Linsenmeier et al.\ (2022a). See also Linsenmeier et al.\ (2022b).} Similarly as these authors, we also use a number of country characteristics to analyze countries' actual carbon prices. However, unlike these authors, we focus explicitly on counties' population sizes that the theory predicts to be a key driver of countries' climate policies, if the free-rider hypothesis holds true.
Regarding the effectiveness of carbon prices as a means to regulate emissions, Best et al.\ (2020) provide evidence (using data from 142 countries, over two decades) that the annual growth rate of CO$_2$ emissions is lower in those countries with a carbon price in place, than in those without.

\medskip

The remainder of this paper is organized as follows. Section~\ref{sec:model} lays out our simple theoretical model and provides key predictions from the theory. These provide the backdrop against which countries' carbon prices are analyzed empirically in Section~\ref{sec:empirical}. Specifically, our goal is to test the free-rider hypothesis, based on the (theoretically predicted) positive relation between the population share of a country and the stringency of its climate policy, measured by its implemented carbon price. Section~\ref{sec:conclusion} concludes. Parts of the theoretical and empirical analysis as well as formal proofs are relegated to the Appendix.

\section{Model} \label{sec:model}

We consider a simple model in which the world is divided into $N$ countries. Countries differ in population size, but are otherwise symmetric. We choose this modeling approach in order to isolate any country-size effects under non-cooperative decision-making. While real-world countries are asymmetric in various other dimensions, country-size effects should prevail in empirical analyses, unless the various dimensions in which countries differ are perfectly correlated (which would appear to be implausible).

We further assume, as is frequently done in the game-theoretic literature on non-cooperative climate policy and on climate cooperation, that countries' social welfare is interdependent only via the climate externality that they inflict upon each other (see Nesje 2022 for a treatment of interdependent utilities). Hence, we abstract from other types of interdependencies, such as carbon leakage and international trade.

Denote by $b(\cdot)$ and $c(\cdot)$ the benefit and cost functions. Let the utility of a representative person in country $i$ be
\begin{equation} \label{utility}
u_i(a_1, ..., a_N) = b(A)-c(a_i),
\end{equation}
where $a_i$ denotes the abatement of emissions per capita in country $i$, and $A$ is the aggregated abatement of all countries.\footnote{The emission of greenhouse gases is a transboundary environmental problem, i.e., climate damages in a country depend primarily on the aggregated emissions of all countries.} Let $P_i$ be the population size of country $i$. The total abatement of country $i$ is then given by $A_i \equiv P_i \cdot a_i$, and the aggregated abatement of all countries is
\[
A = \sum_{j=1}^{N} A_j = \sum_{j=1}^{N} P_j a_j.
\]

In the following we assume that countries differ only in their population sizes ($P_i$), but are otherwise ex ante identical. Our goal is to analyze how the population size of country $i$ affects country $i$'s climate policy under non-cooperative decision-making, as well as in a situation where the government in country $i$ seeks to maximize global rather than national social welfare. Additionally, we will endogenize the formation of a climate coalition to show that under (partial) climate cooperation, the country-size effect on carbon prices may be exacerbated, compared to the fully non-cooperative case.

\subsection{Non-cooperative decision-making} \label{sec:non:coop}

We begin with the case of non-cooperative decision-making. This is usually the ``standard'' case analyzed in the game-theoretic literature on climate policy. The idea is, that the government (or decision-maker) in country $i$ seeks to maximize the social welfare of the population of country $i$, thereby neglecting the social welfare of people living in other countries. In the context of our model, this means that the government in country $i$ neglects the climate externality of country $i$'s emissions on other countries. 

Formally, the decision-maker in country $i$, thus, seeks to maximize the utility of the representative person in country $i$:
\[
\max_{a_i} u_i(a_1,...,a_N) = b(A_{-i} + P_i a_i) - c(a_i),
\]
where $A_{-i} \equiv \sum_{j \neq i} A_j$ is the abatement of all other countries, and country $i$'s own abatement is $A_i=P_i \cdot a_i$. Assuming that $b(\cdot)$ is (weakly) concave, and $c(\cdot)$ is strictly convex,\footnote{At least one of these functions has to be non-linear to assure existence of an interior solution.} the first-order condition associated with the above maximization problem is
\begin{equation} \label{eq:non:coop}
P_i \cdot b'(A) = c'(a_i).
\end{equation}
In a non-cooperative Nash equilibrium, condition~\eqref{eq:non:coop} has to hold for every country $i$ individually. Hence, this constitutes a set of $N$ equilibrium conditions, one for each of the $N$ countries.

Assuming that each of the $N$ countries regulates its emissions with the help of a carbon price, let us denote country $i$'s unilateral carbon price by $\tau_i$. Since individuals (by assumption) do not internalize the climate externalities of their emissions, the decision-maker in country $i$ has to set $\tau_i$ such that $\tau_i = c'(a_i)$ to implement the optimal abatement level $a_i$ (as defined by \eqref{eq:non:coop}). Hence, we obtain that
\begin{equation} \label{tau:non:coop}
\tau_i = P_i \cdot b'(A).
\end{equation}
This insight is key when establishing the following proposition:
\begin{proposition} (Nordhaus, 2015) \label{prop:non:coop}
The level of the non-cooperative carbon price in country $i$ is proportional to the population size of country $i$, $P_i$.\footnote{Nordhaus (2015) establishes in a related framework that ``a country's noncooperative carbon
price is equal to the country share of output times the global social cost of carbon''. He notes that this result ``survives alternative specifications of the abatement-cost function'', but does not clarify under what condition(s) it holds more generally, beyond the specification that he uses (with linear climate damages and quadratic abatement costs). We offer a general condition.}
\end{proposition}

It follows directly from our assumption, that abatement externalities are internalized within the population of country $i$ (but neglected across borders). Although the result is straight-forward and, thus, not very surprising, it has important implications. In particular, it establishes a sharp theoretical connection between the size of a country (specifically its population size), and the stringency of the country's climate policy under non-cooperative decision-making, in particular the level of the carbon price that this country would implement. If the free-rider hypothesis is correct, then such a country-size effect should be visible in the data. We will return to this issue further below.

Another interesting implication of Proposition~\ref{prop:non:coop} is, that when the non-cooperative carbon prices of two countries are compared with each other, their {\em ratio} equals the ratio of the population sizes of these countries (see \eqref{tau:non:coop}):
\[
\frac{\tau_i}{\tau_j} = \frac{P_i}{P_j}.
\]
The following corollary describes this result:
\begin{corollary} \label{corr:1}
The ratio of the non-cooperative carbon prices of two different countries is {\em independent} of the specification of the benefit and cost functions of abatement, and equals the ratio of the population sizes of these countries.
\end{corollary}

That we find a clear ordering of countries' non-cooperative carbon prices, which is independent of the underlying benefit and cost functions, makes the predictions of the theory particularly robust. This is a useful property that allows us to test the free-rider hypothesis, without relying on specific functional forms or estimates of marginal abatement costs.\footnote{For illustration, consider three countries with similar GDP per capita, but very different population sizes: the US (331 million in 2020), Japan (126 million), and Sweden (10 million) (source: https://www.worldometers.info/population/world/, visited Feb 05, 2022.)
By Corollary~\ref{corr:1}, the theory predicts
$\frac{\tau_{US}}{\tau_{Japan}} \approx 2.6 \, ,  \,
\frac{\tau_{US}}{\tau_{Sweden}} \approx 33$.
Based on data from Dolphin (2022), the emission-weighted carbon prices implemented in 2020 in these countries were approximately (in US-Dollars): 0.73 in the US, 1.93 in Japan, and 67.4 in Sweden, so the actual data reveal the {\em opposite} ordering of carbon prices than predicted by the theory:
$\frac{\tau_{US}}{\tau_{Japan}} \approx 0.38 ,
\frac{\tau_{US}}{\tau_{Sweden}} \approx 0.01$.}
While estimates of the benefit and cost functions will vary across countries in practice, country-size effects (as predicted by the theory) should in aggregate, across a large set of countries, nevertheless prevail in the data, if the free-rider hypothesis is valid. 

\subsection{Global social welfare oriented policy-making}

Now suppose, that the decision-maker in country $i$ seeks to maximize global social welfare when deciding about the climate policy of country $i$. For the time being, we leave open the question, if the decision-makers in the other countries are also global social welfare oriented in their decision-making, or if they focus on domestic social welfare (as we assumed before). The only assumption that we make in the following is, that at least within {\em one} country (here: country $i$), the decision-maker is global social welfare oriented.

Let us denote global social welfare by
\[
W(a_1,...,a_N) \equiv \sum_{j=1}^{N} P_j \cdot u_j(a_1,...,a_N).
\]
Then the maximization problem of the global social welfare oriented decision-maker in country $i$ reads
\[
\max_{a_i} W(a_1,...,a_N) = \sum_{j=1}^{N} \big[ P_j \cdot b(A_{-i} + P_i \cdot a_i) \big] - P_i \cdot c(a_i) - \sum_{j \neq i} P_j \cdot c(a_j) ,
\]
where we have isolated the impact of the abatement (per capita) within country $i$, $a_i$, upon social welfare. Notice, that the aggregated abatement: $A=A_{-i} + P_i \cdot  a_i$, does not depend on the summation index $j$ in the first term on the right-hand side of the above equation. Hence, this term can be written more conveniently as $b(A) \cdot \sum_{j=1}^{N} P_j$, which is simply $b(A) \cdot P$, if we denote by $P \equiv \sum_{j=1}^{N} P_j$ the aggregated (global) population size.

The first-order condition associated with the above maximization problem, thus, reads (after dividing both sides by $P_i$)
\begin{equation} \label{eq:coop}
P \cdot b'(A) = c'(a_i).
\end{equation}
The carbon price that the decision-maker in country $i$ uses to implement the abatement target implied by this condition is
\begin{equation} \label{tau:coop}
\tau_i = P \cdot b'(A).
\end{equation}
This leads us directly to the following result:
\begin{proposition} \label{prop:coop}
If the decision-maker in country $i$ is global social welfare oriented, then the carbon price implemented unilaterally in country $i$, $\tau_i$, only depends on the aggregated world population, $P$, and not on the (relative) population size of country $i$, $P_i / P$.
\end{proposition}

Hence, irrespective of the size of the respective country, any decision-maker who is global welfare oriented would set the {\em same} carbon price (as defined by \eqref{tau:coop}) according to our model. In other words: In contrast to the non-cooperative case, there is {\em no} country-size effect on the domestic carbon price of a country that seeks to maximize global social welfare.

Until now, we have left open the question, how the decision-makers in the other countries determine their carbon pricing policies, when the decision-maker in country $i$ is global social welfare oriented. For simplicity, let us illustrate how this affects the overall outcome of all countries by considering two extreme cases. The first case is where decision-makers in {\em all} countries behave in this way and, thus, each of them has global social welfare in mind when determining the climate policy for their own country. It is easy to see, that the {\em global social welfare optimum} is then attained.
The overall abatement in this optimum, denoted by $A^o$, solves (using \eqref{eq:coop})
\[
P \cdot b'(A^o) = c'(A^o/P).
\]
In each country, an identical carbon price then prevails that is (by \eqref{tau:coop}) equal to
\begin{equation} \label{tau:global:opt}
\tau^o = P \cdot b'(A^o).
\end{equation}
The per capita abatement target in the global welfare optimum is, then, also identical across all countries, and given by (using \eqref{eq:coop}): $a^o = c'^{-1} (\tau^o)$.

The second case that we consider is where decision-makers in all {\em other} countries behave non-cooperatively, while the decision-maker in country $i$ seeks to maximize global social welfare.\footnote{This case is interesting especially in the context of the findings in our companion paper (Drupp, Nesje, Schmidt 2022), where experts on aggregate recommend higher unilateral carbon prices than at the global level. Our theoretical analysis in the present paper may help to shed light on why.} For simplicity, we refer to this as the case with ``mixed motives'' of the decision-makers. To fix ideas, let us (without loss of generality) assume that the global welfare oriented policy-maker is in country 1 (hence, $i=1$), so that countries $2,...,N$ behave non-cooperatively. Then the overall outcome ($a_1,a_2,...,a_N$) is characterized by the following set of conditions (by \eqref{eq:coop} and \eqref{eq:non:coop}):
\begin{equation} \label{cond:mixed:motives}
P \cdot b'(A) = c'(a_1) , \, \mbox{ and } \, P_j \cdot b'(A) = c'(a_j), \mbox{ for } j=2,...,N,
\end{equation}
where $A=\sum_{j=1}^{N}P_j a_j$. This yields the following carbon prices (by \eqref{tau:coop} and \eqref{tau:non:coop}):
\begin{equation} \label{prices:mixed:motives}
\tau_1 = P \cdot b'(A) , \, \mbox{ and } \, \tau_j = P_j \cdot b'(A), \mbox{ for } j=2,...,N.
\end{equation}
Clearly, the global social welfare oriented decision-maker in country 1 sets a higher carbon price than in all other countries. Furthermore, by comparing \eqref{prices:mixed:motives} with \eqref{tau:global:opt}, we establish the following result:

\begin{proposition} \label{prop:mixed:motives}
If the benefit function $b(\cdot)$ is strictly concave, then in the case with mixed motives, the global social welfare oriented decision-maker in country 1 sets a carbon price that is {\em higher} than the carbon price in the global social welfare optimum.
\end{proposition}

\subsection{Climate cooperation} \label{sec:coop}

An influential strand of literature (e.g., Barrett 1994, Karp and Simon 2013) analyzes the formation of a climate coalition that endogenizes the environmental externality between its members. The standard assumption is that individually, each country seeks to maximize the social welfare of its own population (see Section~\ref{sec:non:coop}), but once inside the coalition, all members behave fully cooperatively amongst each other. Each country's participation decision in the coalition, however, remains non-cooperative. As Barrett (1994) shows, countries' free-rider incentives then typically translate into a small number of participating countries, especially if the potential benefits from cooperating are large. This theoretical prediction has been shown to be robust across various specifications of countries' benefit and cost functions, as well as for different modifications of the basic modeling framework.\footnote{E.g., see Barrett (2001, 2006, 2013), Finus and Maus (2008), among others.} 
The goal here is to extend Barrett's (1994) analysis that focuses on ex ante symmetric countries to a setting where countries differ in their population sizes.\footnote{Other authors (e.g., Finus and McGinty 2019) have analyzed asymmetries in the context of climate coalition formation, but to the best of our knowledge, the case where countries differ only in their population size and are otherwise symmetric, has not been analyzed before.} It turns out that the analysis of this case is simple, but yields surprisingly sharp results. As we will show below, these theoretical results once more point in the direction of larger countries having more ambitious climate policies and, thus, higher carbon prices. 

We focus on the following utility function of a representative person in country $i$:
\begin{equation} \label{ut:lin:quadr}
u_i(a_1, ..., a_N) = \beta A-\frac{\gamma}{2}a_i^2.
\end{equation}
Hence, we adopt a simple linear-quadratic specification of the utility in \eqref{utility}, as is often used in the literature (e.g., Battaglini and Harstad 2016). 

In line with Barrett (1994), the game consists of two stages: In stage 1, all countries simultaneously decide whether or not to join the coalition (each of them individually and non-cooperatively). In stage 2, the members of the coalition that formed in stage 1 set their abatement targets cooperatively, so as to maximize their joint social welfare, whereas each of the outsiders sets their abatement target non-cooperatively.\footnote{Given the linear benefits that we assume here, the abatement choices of the coalition members resp.\ of the outsiders can also be treated as two separate stages. This leads to identical results.
} The population size of country $i$ is again denoted by $P_i$, and countries are ordered by their population sizes so that country 1 has the largest and country $N$ the smallest size.\footnote{We rule out by assumption that any two countries have identical population sizes.}

The game is analyzed formally in Appendix~\ref{sec:App:clim:coop}. The following proposition summarizes the main results from this analysis:

\begin{proposition} \label{prop:coalition}
With asymmetric population sizes, in any pure strategy Nash equilibrium of the climate cooperation game, at most two countries join the coalition, and the largest country is always one of them. If $2P_2 < P_1$, country 1 forms a singleton coalition. If $2P_2 = P_1$, either a singleton coalition forms, or a coalition of countries 1 and 2. If $2P_2 > P_1$, the equilibrium coalition size is equal to two. Apart from country 1, another country from the set $\{ \mbox{country 2, country 3, ... , country } l \}$ joins, where country $l$ is the smallest country satisfying $P_l \geq \max \{ \frac{P_1}{2},2P_2-P_1\}$.
\end{proposition}

This implies that under asymmetric population sizes (and otherwise identical countries), the equilibrium outcome is either unique, or there is some multiplicity of equilibria, reflecting an indeterminacy of the identity of the ``other country'' that joins the coalition, apart from country 1. However, as we argue in the Appendix~\ref{sec:App:clim:coop}, the degree of multiplicity is typically small, and the natural equilibrium candidate (or focal point) is that the {\em largest two countries} form a coalition if $2P_2 > P_1$ holds.

We conclude that also when countries can form a climate coalition, we should expect higher carbon prices in bigger countries, and lower carbon prices in smaller ones. The formation of a coalition of two large countries (the largest one, and typically the second-largest, or a country with almost the same population size as in the second-largest country) exacerbates our earlier prediction that under non-cooperative policy-making, we expect to see {\em higher carbon prices in larger countries than in smaller ones}. We will test this prediction in the following section.

\section{Empirical analysis} \label{sec:empirical}

We focus our empirical analysis on whether there is evidence supporting the existence of a country-size effect in real world carbon pricing data. To do so we use the emission-weighted average of sector(-fuel) price data (ECP) from Dolphin (2022),\footnote{The prices are the existing total prices including any potential rebate and are expressed in 2019USD/tCO2e. Dolphin (2022) uses verified emissions data to calculate the shares of sector-level (or sector-fuel-level) in total jurisdiction (national or subnational) CO$_2$ emissions. Additionally, the ECP is calculated separately for carbon taxes and emission trading schemes (ETSs) and the combination of both. We use the national CO$_2$ prices for 2020 for the combination of carbon taxes and ETSs, as multiple jurisdictions covered contain both kind of instruments.} as well as population sizes from the World Bank's World Development Indicators data set (World Bank 2022). For each country we make use of the most recent accessible data in each data set, with the vast majority of values from 2020.

We start by investigating the relationship between the emission-weighted carbon prices and the population size of 195 countries via linear regression analysis, assuming that a nonexistent emission-weighted carbon price is treated as a price of \$0.
Across the considered theoretical models, the free-rider hypothesis predicts a positive association between the two. We find, however, that there is no significant univariate relationship between ECP and population size by using ordinary least squares for estimating the coefficients of a linear regression model.  This result holds true if we analyze the complete data set, if we omit the observations from the two countries with the largest population sizes (China and India),\footnote{The population sizes of these countries exceed those of all other countries drastically and can, thus, be seen as outliers that tend to affect results more strongly than observations from other countries. Omitting these observations is, thus, a simple way to control for their impact on the results.} and also if we only consider European countries. For each of these three sets, we also consider the subsets containing only the countries with existing carbon prices.

\begin{figure}[h]
    \centering
    \includegraphics[width=\textwidth]{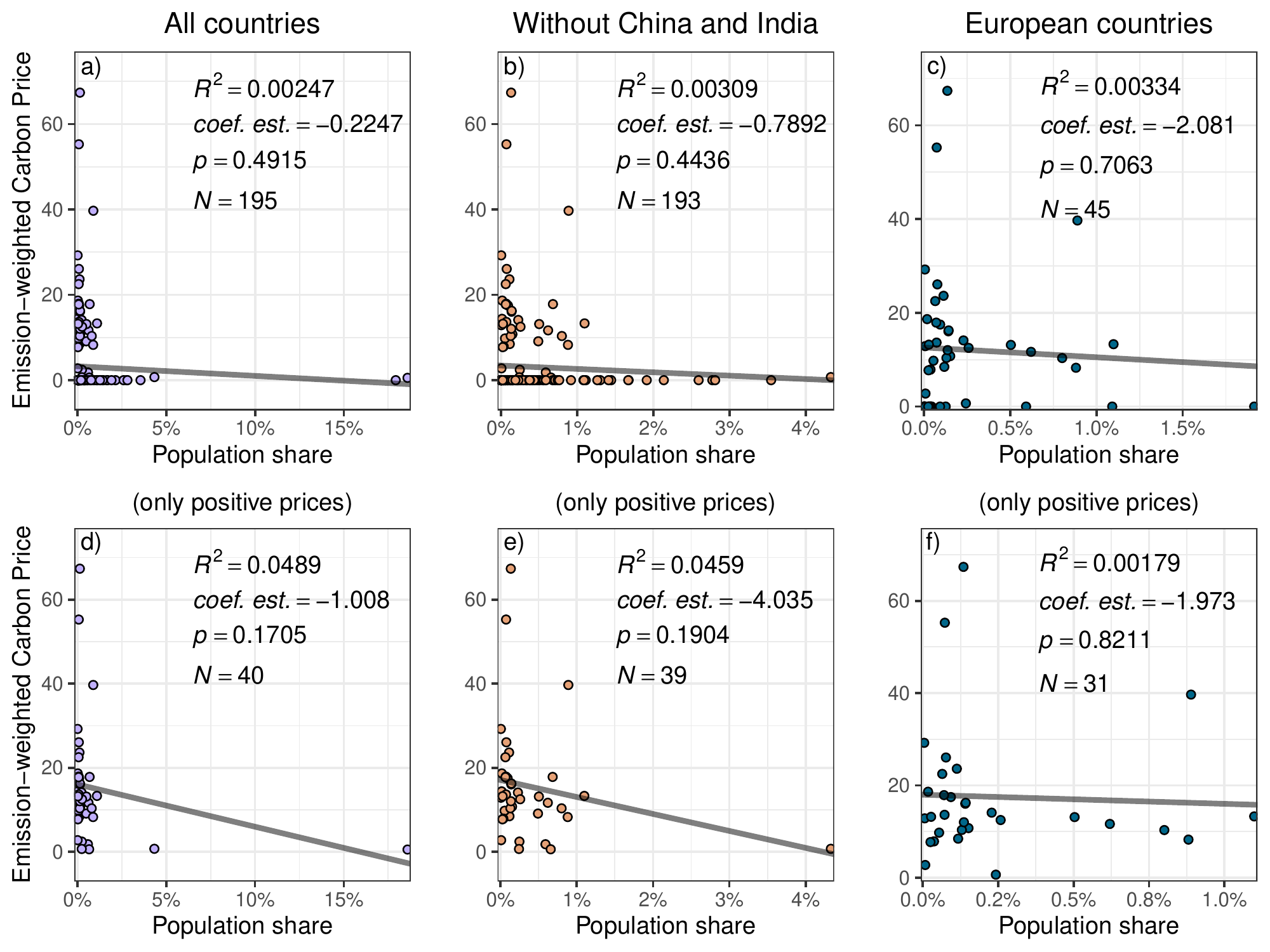}
    \caption{Scatterplots with linear regression lines and summary statistics for different subsets. Panel a) to c) consider different sets of countries including those without carbon prices, coded as 0\$. Panel d) to f) consider the same sets of countries as above, but omitting all observations with prices of \$0.}
    \label{fig:regressions}
\end{figure}

We summarize the results of the regression models in Figure \ref{fig:regressions}. Panel a) contains the result of the regression for the complete data set, containing 195 countries. It displays a scatter plot, with the carbon price values on the vertical and countries' shares of the global population (in percent) on the horizontal axis. In addition, the grey line indicates the fitted regression line and in the top right corner are values summarizing the regression model, including the coefficient of determination $R^2$, the coefficient estimation for the variable population $\beta$, the corresponding $p$-value, and the sample size $N$. In the first case, considering all countries, we cannot observe a positive linear relationship between ECP and population share. The estimated coefficient for population share is insignificant ($\beta=-0.22$; $p=0.49$).
Panel b) shows the result of excluding the two countries with the largest population share, China and India, and is similar to the first result, with a very low $R^2$ and a statistically insignificant coefficient.
Panel c) on the top right of Figure \ref{fig:regressions} displays the results of the regression based on the much smaller subset of all 45 European countries. Also in this case, there is no evidence for a positive linear association between the carbon prices and the population share.
Results are similar when considering only the individual subsets of countries with existing (positive) carbon prices, illustrated in the lower row of Figure \ref{fig:regressions} in panels d) to f). Neither for all countries (with and without China and India) nor for the European countries do we find a significant a correlation between carbon prices and population size.


We proceed by analyzing the relevant subset of countries with existing carbon prices to tackle the problem of the zero-inflated emission-weighted carbon price data. We examine a variety of linear regression models including a number of other key explanatory variables that may also explain the levels of carbon prices, and investigate the relationship between population size and carbon prices. These additional variables stem from different sources, with most of them gathered from the World Bank's World Development Indicators data set (World Bank 2022). In addition to the country's share of the global population, this includes: CO$_2$ emissions per capita, the mean of World Governance Indicators, the share of fossil fuels in the country's energy consumption, and the share of the industry sector in GDP. In addition to that, we consider data on ambient air pollution from the World Health Organisation's Global Health Observatory Data Repository (World Health Organization 2022), data on climate change awareness from the Gallup World Poll (Gallup World Poll 2007/2008), and data on expected climate damages per capita derived from Ricke et al. (2018).\footnote{Since there are missing values in almost all of the considered data sets, we used multiple imputation technique (Azur et al. 2011) to estimate missing values to minimize the bias that could arise from excluding observations with missing values, as these observations are often from small, low-income countries. As a robustness check, we also specified all reported models using case-wise deletion instead of multiple imputation and found no major differences between the two (see Appendix~\ref{missing_values}).} More information about the data sets and the data gathering process can be found in Appendix \ref{data collection}.

To avoid the potential problem of multicollinearity, we calculated the cross-correlation matrix of all independent variables and exclude those combinations of variables with correlation coefficients higher than 0.75. This results in removing the variable of climate change awareness (see Appendix~\ref{cross-correlation}). However the reported results are robust when adding climate change awareness as an additional independent variable.\footnote{Furthermore, we conducted a variance inflation factor (VIF) test for each of the reported model specifications, and obtain no values larger than $4$ across all specifications and variables.}

The specification chart in Figure~\ref{fig:spec} showcases a series of estimates for the coefficient of population size from a wide range of models. Each column corresponds to a different linear regression model.

\begin{figure}[h]
    \centering
    \includegraphics[width=\textwidth]{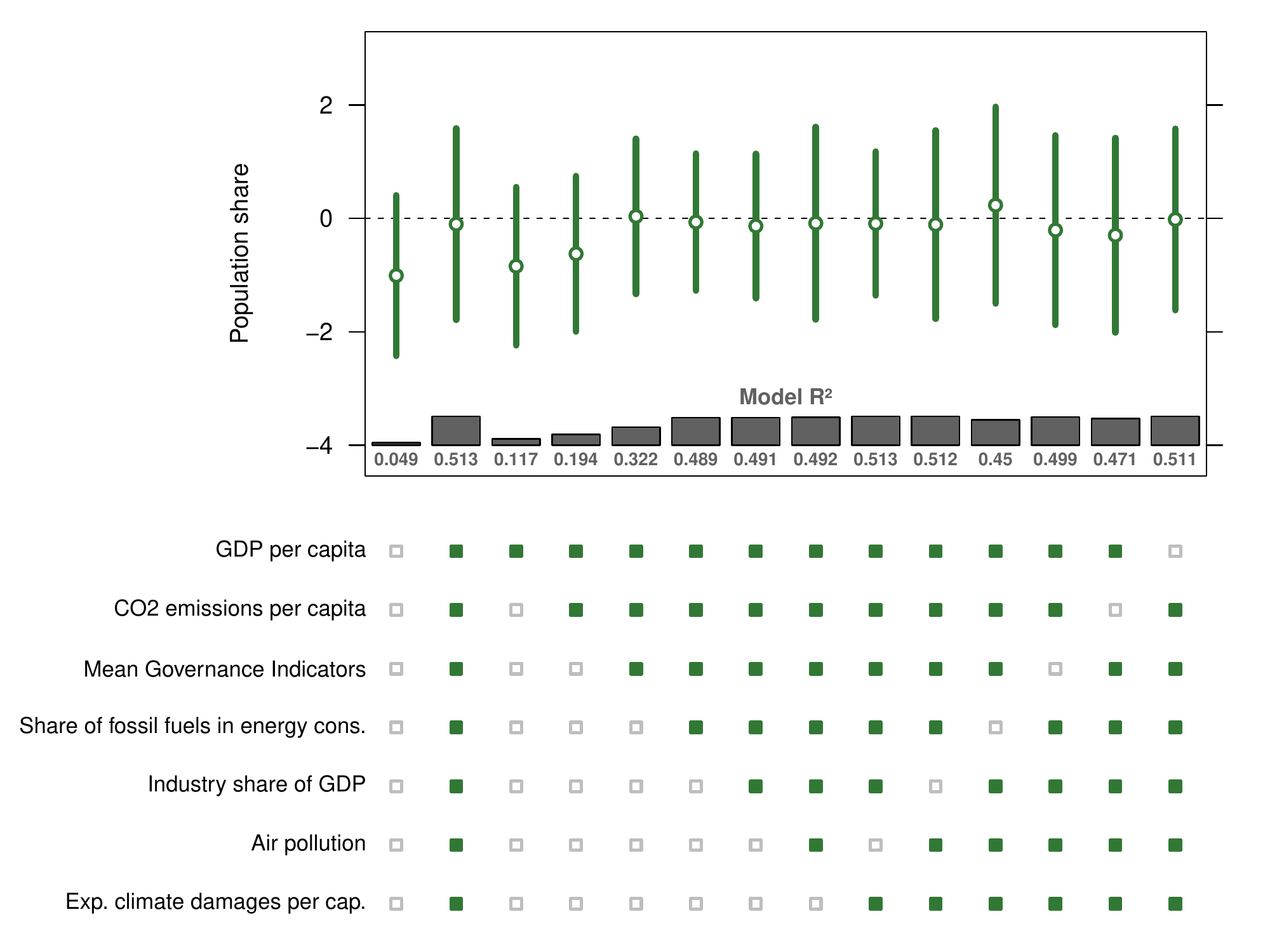}
    \caption{Specification chart for the effect of population shares on carbon prices. Coefficient estimates with confidence intervals of population share for different specifications of linear regression models on imputed data sets (N=40). Column 1 shows the effect of population share for the univariate linear regression model containing only population share as an independent variable. Column 2 shows the effect of population share on carbon prices for the “full" model, containing all selected explanatory variables. Columns 3 to 7 build up upon the univariate model by adding individual variables and columns 8 to 12 show the results for the “full but leave-one-out" model.}
    \label{fig:spec}
\end{figure}

While the various input variables differ between the models, the share of global population as a measure for population size is an explanatory variable contained in all of them. Column 1 corresponds to the simple model regressing emission-weighted carbon prices only on population size with the same results of negative but insignificant coefficient estimation as already discussed for Figure \ref{fig:regressions}. The $R^2$ value of the respective models plotted is given below the coefficient estimates. In the first case this value is 0.049 as already seen in Figure \ref{fig:regressions}.
The lower part of the plot visualizes which explanatory values are part of each regression model. As we can see, column 2 corresponds to the full model, in addition to the population share containing the variables: GDP per capita, CO$_2$ emissions per capita, the mean values of Governance Indicator rank scores, percentage share of fossil fuels in a given countries energy consumption, the share of industry of a countries GDP, ambient air pollution exposure, and expected climate damages per capita. The effect of population share in this setup is not different from zero, meaning that -- when incorporating all additional explanatory variables -- no significant linear relationship between population share and carbon prices can be identified. The next columns (columns 3 to 8) build upon the simple model by adding one input variable after the other. For each of these models the estimated coefficient for population size is slightly negative but insignificant. The last columns (from column 9 on) show results from a ``leave one out approach''. These contain all variables from the full model but one, i.e., leaving one variable out in each case. As before, for all of these models, we find insignificant estimates that are close to zero for the coefficient of population share. To summarize, we do not find a linear relationship between the population size and the emission-weighted carbon prices for the subset of countries with existing carbon prices and over the analyzed models.

\mycomment{
\subsection{Comparing countries with and without a carbon price}
After investigating the subsample of countries that have already implemented carbon prices, we proceed to analyze the difference between countries with established carbon prices and countries without. To do so, we consider logistic regression models, introducing a binary dependent variable with value 1, if a carbon price exists in a given country, and 0 otherwise. With this analysis we test whether any country-size effect can be found between these two groups of countries. To do so, we estimate multiple logistic regressions with the same setup as above. For the complete set of 195 countries, we obtain consistently positive coefficients for all estimated logistic regression models and significant results for most of them (see Appendix \ref{logistic_regression}). However, these results are not robust under removal of the largest observation in terms of population, China, that has a positive but very small carbon price of less than 1\$. After repeating the analysis for the complete data set with this single observation removed, we obtain vastly different results: estimates are no longer statistically significantly different from zero. The same holds true when removing both outliers in terms of population share, China and India (see Appendix \ref{logistic_regression}). The positive relationship between the existence of a carbon price and the population share is therefore not robust and highly influenced by these outliers.
}

\section{Conclusion} \label{sec:conclusion}

We conclude from the lack of empirical evidence of a positive relationship between the population sizes of countries and their established carbon prices, that the free-rider hypothesis cannot be supported based on actual carbon pricing data. If free-riding were among the key obstacles for effective climate policies around the globe, as has been postulated for decades, then according to the canonical way of modeling non-cooperative policy-making, we would expect to see a pronounced country-size effect in the data. This is because a country with a larger share of the global population should internalize more of the climate externalities that accrue inside of its borders than a smaller country. Any efficient climate policy mix that internalizes such externalities comprises a carbon pricing scheme, according to basic economic principles. Therefore, larger countries would be expected to establish higher carbon prices than smaller countries, if free-riding indeed plays a major role. Controlling for various other country characteristics, a positive country-size effect on carbon prices is, however, not identified. 
This points at other obstacles that can prevent countries from implementing effective climate policies, and that may be more important in practice than the free-rider incentive. Such obstacles could be competitiveness concerns under unilateral carbon pricing, (other) political economy considerations, and a lack of awareness and public support for effective climate policy measures by the wider public, to name just a few.

From a policy-perspective, this would suggest that measures targeted primarily at the free-rider problem, such as trade sanctions designed to foster participation in a climate coalition or climate club (Nordhaus 2015) might be less suitable or less effective to address the climate problem than other measures. For example, if competitiveness concerns are a key obstacle for the implementation of effective climate policy measures, then border carbon adjustment may be more effective in tackling the problem than trade sanctions. If instead a lack of public support for climate policy is a key concern, then informational campaigns and investments in education may be more effective. A combination of such approaches is likely to be most effective to curb emissions globally in the next few years and decades in an effort to stabilize the climate system. 
Future research should help to identify what obstacles are most severe, if it is not free-riding (as our results suggest), so that an adequate policy mix can be designed to address the climate problem effectively. 


\newpage

\appendix

\section{Appendix}
\label{sec:Appendix}

\subsection{Appendix for Section~\ref{sec:coop} (climate cooperation)} \label{sec:App:clim:coop}

Consider stage 2 of the game. Hence, a coalition $S$ has already formed. Let $A_n$ be the aggregated abatement of all outsiders, whereas $A_s$ is the total abatement of all coalition members. The total population size of all countries in $S$ is denoted by $P_s$. The social welfare of the coalition is thus given by
\[
W_s = \beta \cdot P_s \cdot (A_n+A_s) - \frac{\gamma}{2} \sum_{i \in S} P_i \cdot a_i^2.
\]
The coalition $S$ can be thought of as a social planner who sets the abatement targets of its member countries $i \in S$, so as to maximize $W_s$. Hence, by joining the coalition, country $i$ gives up its sovereignty regarding climate policy, and delegates this decision to the planner. It is easy to see, that this planner will set identical abatement targets per capita for each country $i$ in the coalition, that is $a_i \equiv a_s$ for all $i \in S$.\footnote{Suppose to the contrary that the per capita abatement targets of two countries inside the coalition were different. Then by keeping the total abatement of the coalition, $A_s$, constant, $W_s$ could be increased by shifting some of the abatement from a country with a higher abatement per capita target to one with a lower target. This follows immediately from the convexity of the abatement cost function $c(a_i)$.} Therefore, the planner's optimization problem reads:
\begin{equation} \label{us}
\max_{a_s} P_s \cdot \left[ \beta \cdot (A_n+P_s \cdot a_s) - \frac{\gamma}{2 \cdot }a_s^2 \right] = P_s \cdot u_s.
\end{equation}
Each person inside the coalition obtains an identical utility of $u_s$.\footnote{It, thus, seems natural in this setting to assume that no transfers take place inside the coalition, as these would necessarily lead to heterogeneous utility levels inside the coalition (given that the planner maximizes total welfare of all coalition members, as described above).}

Solving the above maximization problem, we obtain the first-order condition:
\[
\beta \cdot P_s = \gamma \cdot a_s,
\]
from which we obtain the uniform per capita abatement target of all countries inside the coalition, as well as the total abatement of the coalition:
\begin{equation} \label{as}
a_s = P_s \cdot \frac{\beta}{\gamma} \mbox{ , } A_s = P_s^2 \cdot \frac{\beta}{\gamma} .
\end{equation}
Using this result in \eqref{us}, the resulting utility of each person inside the coalition becomes (after simplifying):
\begin{equation} \label{us:solution}
u_s = \beta \cdot A_n + P_s^2 \cdot \frac{\beta^2}{2 \gamma}.
\end{equation}
The social welfare maximization problem of each outsider country yields a dominant strategy. Using \eqref{eq:non:coop}, we obtain that as an outsider, country $i$ sets an abatement target per capita, resp.\ a total abatement target of:
\[
a_{i,n} = P_i \cdot \frac{\beta}{\gamma} \mbox{ , } A_{i,n} = P_i^2 \cdot \frac{\beta}{\gamma}.
\]
Note, that this corresponds to the per capita abatement target of a coalition member in \eqref{as}, with the only difference that coalition members take the total population size of the coalition in consideration, whereas an outsider only maximizes the welfare of its own population (with size $P_i$).

To complete the equilibrium characterization of this climate cooperation game, it remains to determine, which countries will join the coalition in equilibrium. We restrict our attention to equilibria in pure strategies.\footnote{A climate cooperation game where countries randomize over their participation decisions is analyzed in Hong and Karp (2012).} To this end, we focus on the participation decision of a single country $i$, and treat the participation decisions of all other countries as given. This corresponds to a Nash equilibrium in the participation stage of the game, with countries anticipating the results from the subsequent abatement stage. As a working hypothesis, suppose that there exists a Nash equilibrium where country $i$ joins the coalition on the equilibrium path, i.e., $i \in S$. All we need to do is to compare country $i$'s equilibrium payoff with the payoff that it obtains if it deviates at the participation stage, thereby becoming an outsider. Then the formation of the coalition $S$ is a Nash equilibrium, if no country $i \in S$ has an incentive to deviate in this way, and also no country $j \notin S$ has an incentive to deviate by joining the coalition.

Relevant for country $i$'s participation decision is the utility of a representative person living inside this country. If the country does not deviate and joins the coalition, it is given by $u_s$, see \eqref{us:solution}. If country $i$ deviates and stays outside, then its total abatement equals $A_{i,n}=P_i^2 \cdot \beta / \gamma$. The aggregated abatement of all {\em other} outsiders, that we continue to denote by $A_n$, remains unchanged (i.e., as on the assumed equilibrium path) due to their dominant abatement strategies. The total abatement of all remaining coalition members reduces to (using \eqref{as}):
\begin{equation} \label{As-i}
A_{s,-i} = P_{s,-i}^2 \cdot \frac{\beta}{\gamma},
\end{equation}
where $P_{s,-i} \equiv P_s - P_i$ is the total population of all remaining coalition members. Under the deviation strategy, country $i$ thus obtains a social welfare per capita of
\[
u_{i,n} = \beta \cdot (A_n + A_{i,n} + A_{s,-i}) - \frac{\gamma}{2} \cdot a_{i,n}^2 = \beta \cdot (A_n + A_{s,-i}) + P_i^2 \cdot \frac{\beta^2}{2\gamma}.
\]
The deviation is unprofitable for country $i$ if $u_s \geq u_{i,n}$. This yields the following equilibrium condition:
\[
\beta A_n + P_s^2 \cdot \frac{\beta^2}{2 \gamma} \geq \beta \cdot (A_n + A_{s,-i}) + P_i^2 \cdot \frac{\beta^2}{2\gamma}.
\]
Note, that $\beta \cdot A_n$ drops out. This is because the abatement choices of all other outsiders remain unaffected by country $i$'s participation decision. In this setup, country $i$'s only incentive to join the coalition is to induce the {\em other coalition members} to increase their abatement efforts, since by joining, the coalition obtains a higher total population size. Using \eqref{As-i}, the condition then becomes:
\[
P_s^2 \cdot \frac{\beta^2}{2 \gamma} \geq P_{s,-i}^2 \cdot \frac{\beta^2}{\gamma} + P_i^2 \cdot \frac{\beta^2}{2\gamma}.
\]
Dividing both sides by $(\beta^2/2\gamma)$, and replacing $P_s$ by $P_{s,-i}+P_i$, the equilibrium condition becomes:
\[
(P_{s,-i}+P_i)^2  \geq 2P_{s,-i}^2  + P_i^2.
\]
This simplifies to the simple but crucial condition:
\begin{equation} \label{eq:cond}
2P_i \geq P_{s,-i}.
\end{equation}
In words: Country $i$ has an incentive to join the coalition, if and only if the aggregated population size of all {\em other} coalition members is not larger than twice the population size of country $i$.

To shed further light on this result, let us briefly revisit the benchmark case where all $N$ countries are ex ante symmetric (see Barrett 1994). It is well-known that under symmetry, in the linear-quadratic case, there are two ``stable'' coalition sizes: $k^*=2$ and $k^*=3$ (see also Barrett 1995). However, under symmetry, there is a multiplicity of equlibria also with respect to the identities of the countries that join the coalition. Of course, if countries coordinate, say, on an equilibrium with $k^*=3$ coalition members, the identities of these three countries are not uniquely determined.

Going back to condition \eqref{eq:cond}, with symmetric countries, the condition is satisfied with equality if there are two other countries in the coalition (apart from country $i$). Hence, a third country is indifferent between joining and staying out, thus replicating the standard result from the literature that two or three countries may join. Formally, under symmetry, $k^*=3$ is the maximum coalition size that satisfies ``internal stability'' (i.e., no insider country wants to deviate and stay outside of the coalition). And $k^*=2$ is the smallest coalition size that still satisfies ``external stability'' (i.e., no outsider country wants to deviate and join the coalition).

However, under asymmetric population sizes, such as considered here, condition \eqref{eq:cond} yields much sharper results, both with respect to the equilibrium coalition size, as well as with respect to the identities of the countries that join the coalition. These are summarized in Proposition~\ref{prop:coalition} in the main text.  

%
%

\medskip

We now argue why (as claimed in the main text after Proposition~\ref{prop:coalition}) the degree of multiplicity of equilibria is typically small. To this end, consider the case where $P_1-P_2 = P_2-P_3 \equiv \Delta$. Hence, the difference in population sizes between the biggest and the second biggest country ($\Delta$), is the same as between the second and the third biggest country. This is a useful benchmark case. Then for country 3, the condition $P_3 \geq 2P_2-P_1$ holds with equality, as is easy to verify. Hence, there are two equilibria: apart from country 1, either country 2 or country 3 may join the coalition. But this is a knife-edge case. If $P_1-P_2 < P_2-P_3$, the unique equilibrium outcome is that countries 1 and 2 form a coalition. If $P_1-P_2 > P_2-P_3$, instead of country 2, country 3 may join the coalition in equilibrium. Intuitively, the multiplicity of equilibria is larger only in the special case where country 1 is bigger than country 2, but there are several countries that are roughly of equal size as country 2. Then it is not uniquely determined, which of these countries joins the coalition.

In spite of this indeterminacy (if it exists), the most natural equilibrium candidate is that the largest two countries form a coalition if $2P_2 > P_1$ holds. This is a natural focal point, as country 2 is more ``inclined'' to join than any smaller country. In this setup, larger countries benefit relatively more from joining than smaller ones, because once inside a coalition, all individuals in the participating countries achieve identical utility. Yet, in a country with a larger population size, there are more individuals who enjoy the benefits from climate cooperation than in a smaller country. In other words, a smaller country is more eager to become a free-rider: if it were to join the coalition, the abatement target (per capita) of this country would increase sharply, whereas the other coalition members raise their efforts only a little bit, due to the country's small population size. Hence, a small country has little to gain but a lot to lose from joining the coalition, whereas a larger country can gain more.

\subsection{Proofs} \label{sec:proofs}

\textit{Propositions 1 and 2} follow directly from the discussions in the main text.

\begin{proof}[Proof of Proposition~\ref{prop:mixed:motives}]
By the convexity of the abatement cost function $c(\cdot)$, the decision-maker in country 1 does not fully compensate for the reduced incentives to abate of the decision-makers in the other countries (see \eqref{cond:mixed:motives}). Hence, the aggregated abatement in the case with mixed motives is strictly smaller than in the global social welfare optimum. The result then follows immediately from the concavity of the function $b(\cdot)$ (using the first condition in \eqref{prices:mixed:motives}).
\end{proof}

\begin{proof}[Proof of Proposition~\ref{prop:coalition}]
All results follow from applying \eqref{eq:cond}. We first show that under asymmetric countries, at most two countries can be coalition members in any pure strategy Nash equilibrium. To see this, start from the benchmark case with symmetric countries and suppose three countries form a coalition. Then \eqref{eq:cond} holds with equality for each of them. Now introduce an asymmetry by making one of the three countries smaller. Then clearly, \eqref{eq:cond} now holds with strict inequality for each of the bigger countries, whereas it is violated for the smaller country. By introducing an asymmetry also among the two bigger countries, \eqref{eq:cond} clearly continues to hold for the biggest country (country 1), and may continue to hold or fail to hold for the second largest country. This proves that at most two countries can be coalition members in any pure strategy Nash equilibrium.

Now consider the case where $2P_2 < P_1$. If country 1 joins the coalition, then \eqref{eq:cond} is violated for each of the countries $2,3,...,N$, so that country 1 and any of the other countries forming a coalition is not an equilibrium. It is also not an equilibrium that a coalition of two countries forms that does not include country 1. Since both of these countries are smaller than country 1, \eqref{eq:cond} would always be satisfied for country 1 so that this country deviates and joins the coalition, but a coalition of three asymmetric countries is never stable, as shown above. Hence, the only pure strategy Nash equilibrium is the formation of a singleton coalition by country 1.

Next consider the knife-edge case where $2P_2 = P_1$. It is easy to see that either a singleton coalition of country 1 forms, or a coalition of countries 1 and 2. This is because \eqref{eq:cond} holds with equality for country 2 if country 1 joins, so that country 2 is indifferent between joining and staying out. No other coalition is stable for the reasons indicated above.

Finally consider the case where $2P_2 > P_1$. For the reasons stated above, only a coalition of two countries can be stable. However, there may be some multiplicity of equilibria, as the identity of the second country that joins (apart from country 1) is not always uniquely determined. For this country, the following two conditions must hold: (i) condition \eqref{eq:cond}, i.e., the country must be at least half as big as country 1, and (ii) if this country is {\em not} country 2, then it must not be profitable for country 2 to deviate and join the coalition, i.e., \eqref{eq:cond} must {\em not} hold with a strict inequality for country 2. Let $l$ be the index of the {\em smallest} country for which both of these conditions are satisfied. Hence, $2P_l \geq P_1$, and $2P_2 \leq P_1+P_l$. Combining these two conditions, we find that country $l$ is the smallest country satisfying $P_l \geq \max \{ P_1/ 2,2P_2-P_1\}$, as stated in the proposition.
Then it follows that apart from country 1, another country from the set $\{ \mbox{country 2, country 3, ... , country } l \}$ joins the coalition. If this set comprises only country 2, then there is a unique equilibrium where countries 1 and 2 join the coalition.
\end{proof}

\subsection{Details on the empirical analysis (Section~\ref{sec:empirical})}
\subsubsection{Data Collection} \label{data collection}
The data used for the empirical analysis was collected from different sources and consequently joined. As a measure of country size, we used the percentage of the total global population based on the de facto definition of population, which counts all residents regardless of legal status or citizenship (World Bank 2022).
Data for carbon prices are the existing total prices including any potential rebate and are expressed in 2019USD/tCO2e from Dolphin (2022). Data on GDP per capita is gross domestic product divided by midyear population measured in current USD (World Bank 2022). The CO$_2$ emissions per capita, measured in metric tons per capita, cover all Carbon dioxide emissions stemming from the burning of fossil fuels and the manufacture of cement (World  Bank 2022). The Worldwide  Governance Indicator data set contains six indicators that combine the views of a large number of enterprise, citizen and expert survey respondents in industrial and developing countries towards governance (World Bank 2022). They represent six dimensions of governance: Voice and Accountability, Political Stability and Absence of Violence/Terrorism, Government Effectiveness, Regulatory Quality, Rule of Law and Control of Corruption. For our regression analysis, we used one combined variable, the mean of the percentile rank scores of each indicator. The share from fossil fuel in energy consumption comprises coal, oil, petroleum, and natural gas products, measured in percent of total (World Bank 2022). The data on industry share of GDP measures the percentage of industry (including construction) of the total GDP (World Bank 2022). Data on exposure to air pollution measures the annual mean levels of fine particulate matter in cities (World Health Organization 2022). Climate Change awareness, describes the percentage of a countries population who are aware of climate change (Gallup World Poll 2007/2008). The last variable is expected climate damages per capita, calculated using estimates for expected damages from the reference case using the constant 3\% discount rate out of “Country-level social cost of carbon" (Rike et al. 2018).
Table \ref{summary_stats} contains summary statistics for the collected variables.

\begin{table}[h] \centering
  \caption{Summary statistics of collected variables}
  \label{summary_stats}
\resizebox{\textwidth}{!}{%
\begin{tabular}{ p{5cm} lccccccc}
\\[-1.8ex]\hline
\hline \\[-1.8ex]
 & \multicolumn{1}{c}{Min} & \multicolumn{1}{c}{Pctl(25)} & \multicolumn{1}{c}{Median} & \multicolumn{1}{c}{Mean} & \multicolumn{1}{c}{Pctl(75)} & \multicolumn{1}{c}{Max} & \multicolumn{1}{c}{NA's} \\
\hline \\[-1.8ex]
ECP & 0.00 & 0.00 & 0.00 & 3.12 & 0.00 & 67.35 & 0 \\
Population share & 0.0001 & 0.03 & 0.11 & 0.52 & 0.38 & 18.59 & 1 \\
CO$_2$ emissions per cap. & 0.03 & 0.78 & 2.63 & 4.24 & 5.99 & 32.42 & 7 \\
Share of fossil fuels in energy consumption & 0.00 & 22.59 & 68.77 & 56.05 & 86.62 & 100.00 & 27 \\
GDP per capita & 239.0 & 1,916.5 & 5,246.1 & 15,293.5 & 16,851.4 & 175,813.9 & 0 \\
Mean Governance Ind. & 1.52 & 25.63 & 46.79 & 48.34 & 69.22 & 97.76 & 0 \\
Industry share of GDP & 5.20 & 18.45 & 23.60 & 25.27 & 30.80 & 59.10 & 8 \\
Air pollution & 5.73 & 13.28 & 22.16 & 27.42 & 37.32 & 93.18 & 8 \\
Climate Change awareness & 20.62 & 46.05 & 61.88 & 63.02 & 84.21 & 98.92 & 76 \\
Expected climate damages per capita & -27.56e-07 & 0.05e-07 & 0.87e-07 & -0.02e-07 & 1.29e-07 & 7.44e-07 & 29 \\
\hline \\[-1.8ex]
\end{tabular}
}
\end{table}

\subsubsection{Missing Values}
\label{missing_values}
From the the 195 covered observations, 90 have at least one missing value. There were 76 cases where data on climate change awareness was missing, 27 cases missing for share of fossil fuels in energy consumption and 8 cases missing for Air pollution and Industry share of GDP. In addition, there were missing values for individual World Governance Indicator percentile rank values. We proceeded to use only the mean values of these rank scores, always based on the existing data, leaving no missing values for the mean Governance Indicator values.\\
All of these missing values are a result of the data collection process, as some databases feature a smaller set of countries than others. We assume that the non-inclusion of certain countries follows non-random considerations, as smaller and less affluent countries are less likely to be included. It is therefore expected, that there exists a connection between the absence of values and these values themselves. Because of that, the missing observations need to be classified as “Missing Not at Random (MNAR)” (Kang 2013).

Considering our small sample size and that Missing at Random is not properly satisfied, row wise deletion as a strategy to deal with missing observations is likely to introduce statistical bias (Kang 2013). Therefore, we make use of multiple imputation (replacement values) to estimate the missing values. Multiple imputation is a strategy that creates for each missing data point, several imputed values, which are each predicted by a slightly different regression model from the other variables in the data set to reflect sampling variability. We implemented the multiple imputation procedure using the R package “mice" (Azur et al. 2011). The package creates multiple imputations for our multivariate missing data. The method is based on Fully Conditional Specification, where each incomplete variable is imputed by a separate model iteratively. We estimated five imputed data sets, which is common practice and the default option of the “mice" imputation package. Afterwards we proceeded with individual regression analyses on each imputed data set and then pooled and reported the pooled results in the main text.

\subsubsection{Cross-Correlation Matrix}
\label{cross-correlation}
This cross-correlation matrix is calculated based upon the data from all countries with existing carbon prices. Because of the high correlation coefficient between the mean of Governance Indicators and climate change awareness variables, the latter is removed as an independent variable for the investigated regression models. It was chosen of the two, because of the large amount of missing values compared to the mean of Governance Indicator values.
\begin{figure}[H]
    \centering
    \includegraphics[width=\textwidth]{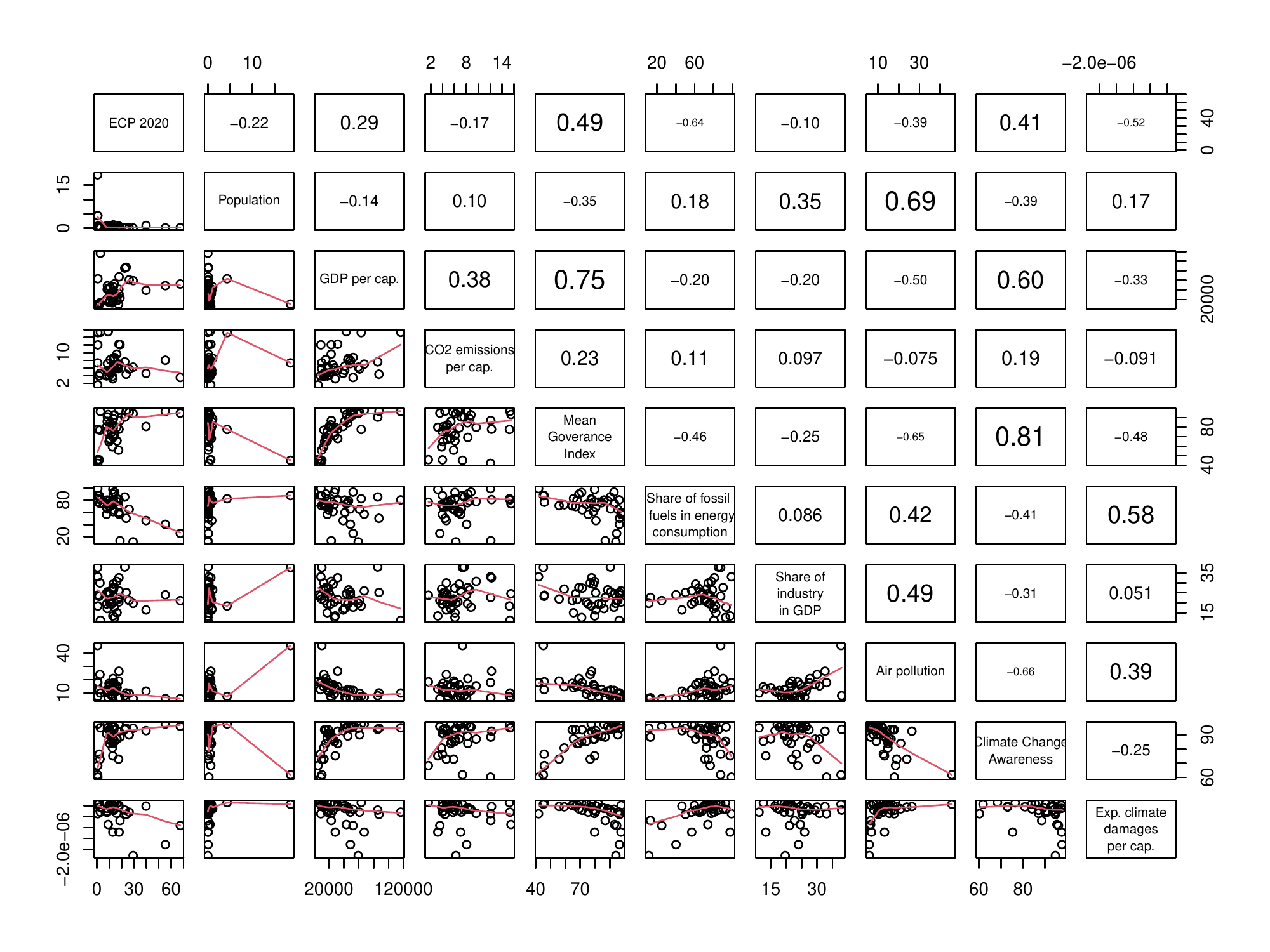}
    \caption{Cross-Correlation Matrix. Pearson correlation coefficient is used.}
    \label{fig:corr}
\end{figure}


\begin{thebibliography}{99}

\bibitem{Azur11} Azur, M. et al. (2011). Multiple imputation by chained equations: what is it and how does it work?
\textit{International journal of methods in psychiatric research}, 20, 40--49.



\bibitem{Bar94} Barrett, S. (1994). Self-enforcing international environmental agreements.
\textit{Oxford Economic Papers}, 46, 878--894.


\bibitem{Bar01} Barrett, S. (2001). International cooperation for sale.
\textit{European Economic Review}, 45,  1835--1850.


\bibitem{Bar05} Barrett, S. (2005). The Theory of International Environmental Agreements.
In \textit{Handbook of Environmental Economics}, vol.\ 3, edited by K.-G.\ M\"aler and J.\ R.\ Vincent. Amsterdam: Elsevier.

\bibitem{Bar06} Barrett, S. (2006). Climate Treaties and ``Breakthrough'' Technologies.
\textit{AEA Papers and Proceedings}, 96,  22--25.

\bibitem{Bar13} Barrett, S. (2013). Climate treaties and approaching catastrophes.
\textit{Journal of Environmental Economics and Management}, 66,  235--250.

\bibitem{BH16} Battaglini, M. and B. Harstad (2016). Participation and Duration of Environmental Agreements.
\textit{Journal of Political Economy} 124, 160--204.

\bibitem{BestBurkeJotzo}
Best, R., Burke, P.J. and F. Jotzo (2020) Carbon Pricing Efficacy: Cross-Country Evidence. \textit{Environmental and Resource Economics}, 77, 69--94.

\bibitem{BestZhang}
Best, R. and Q.Y. Zhang (2020) What explains carbon-pricing variation between countries? \textit{Energy Poicy}, 143, August 2020, 111541.

\bibitem{BoadwayH99}
{Boadway, R. and M. Hayashi} (1999). Country size and the voluntary provision of international public goods. \textit{European Journal of Political Economy}, 15(4), 619--638.

\bibitem{BuchholzS21}
Buchholz, W. and T. Sandler (2021). Global Public Goods: A Survey. \textit{Journal of Economic Literature} 59(2), 488--545.

\bibitem{CS93} Carraro, C. and D. Siniscalco (1993). Strategies for the international protection of the environment.
\textit{Journal of Public Economics}, 52,  309--328.

\bibitem{DixitOlson} Dixit, A. and M. Olson (2000). Does Voluntary Participation Undermine the Coase Theorem?
\textit{Journal of Public Economic Theory}, 76, 309--335.

\bibitem{Dolphin22} Dolphin, G. (2022). Evaluating National and Subnational Carbon Prices: A Harmonized Approach. RFF Working Paper 22-4.

\bibitem{DolphinPollit} Dolphin, G., Pollitt, M.G. and D.M. Newberry (2020). The political economy of carbon pricing: a panel analysis. \textit{Oxford Economic Papers}, 72, 1--29.

\bibitem{PricingCarbon}
Drupp, M.A., Nesje, F. and Schmidt, R.C. (2022). Pricing Carbon. CESifo working paper 2022-03.

\bibitem{Fin08} Finus, M. (2008). Game theoretic research on the design of international environmental
agreements: insights, critical remarks, and future challenges.
\textit{International Review of Environmental and Resource Economics}, 2,  29--67.

\bibitem{FM08} Finus, M. and S. Maus (2008). Modesty May Pay!
\textit{Journal of Public Economic Theory}, 10,  801--826.

\bibitem{FinusMcGinty} Finus, M. and M. McGinty (2019). The anti-paradox of cooperation: Diversity may pay! \textit{Journal of Economic Behavior and Organization}, 157, 541--559.

\bibitem{Fuentes}
Fuentes-Albero, C., Rubio, S.J. (2010). Can international environmental cooperation be bought? \textit{European Journal of Operational Research} 202, 255--264.

\bibitem{Haensel20}
H\"ansel, M.C., Drupp, M.A., Johansson, D.J.A., Nesje, F., Azar, C., Freeman, M.C., Groom, B.and T. Sterner (2020). Climate economics support for the UN climate targets. \textit{Nature Climate Change}, 10, 781--789.





\bibitem{Hoel92} Hoel, M. (1992). International Environmental Conventions: The Case of Uniform Reductions of Emissions.
\textit{Environmental and Resource Economics}, 2, 141--159.

%
%
%

\bibitem{Kang13} Kang, H. (2013). The Prevention and Handling of the Missing Data.
\textit{Korean Journal of Anesthesiology}, 64, 402--406. 

\bibitem{KS13} Karp, L.S. and L. Simon (2013). Participation games and international environmental agreements: A non-parametric model.
\textit{Journal of Environmental Economics and Management}, 65, 326--344.

\bibitem{Kolstad10}
{Kolstad, C.D.} (2010) Equity, Heterogeneity and International Environmental
Agreements. \textit{The B.E. Journal of Economic Analysis \& Policy} 10,
symposium, art.\ 3

\bibitem{KostadT} Kolstad, C.D. and M. Toman (2005). The Economics of Climate Policy.
In \textit{Handbook of Environmental Economics}, vol.\ 3, edited by K.-G. M\"aler and J.R. Vincent, Ch.~30, 1561--1618. Amsterdam: Elsevier.

\bibitem{Korneketal}
Kornek, U., Flachsland, C., Kardish, C., Levi, S. and Edenhofer, O. (2020). What is important
for achieving 2$^\circ$C? UNFCCC and IPCC expert perceptions on obstacles and response options
for climate change mitigation. \textit{Environmental Research Letters}, 15(2), 024005.


\bibitem{Linsenmeier1}
Linsenmeier, M., Mohommad, A. and G. Schwerhoff (2022) The International Diffusion of Policies for Climate Change Mitigation. IMF Working Papers, No. 2022/115.

\bibitem{Linsenmeier2}
Linsenmeier, M., Mohommad, A. and G. Schwerhoff (2022) Policy Sequencing Towards Carbon Pricing -- Empirical Evidence From G20 Economies and Other Major Emitters. IMF Working Papers, No. 2022/066.

\bibitem{Loeper17}
{Loeper, A.} (2017) Cross-border externalities and cooperation among representative democracies. \textit{European Economic Review} 91, 180--208. 


\bibitem{McGinty07}
{McGinty, M.} (2007) International environmental agreements among asymmetric nations. \textit{Oxford Economic Papers} 59, 45--62. 

\bibitem{Nesje22}
{Nesje, F.} (2022) Cross-dynastic intergenerational altruism. CESifo Working Papers, No. 9626.

\bibitem{Nordhaus15}
Nordhaus, W.D. (2015). Climate Clubs: Overcoming Free-riding in International Climate Policy. \textit{American Economic Review} 105, 1339--1370.

\bibitem{Nordhaus19}
Nordhaus, W.D. (2019). Climate change: The ultimate challenge for Economics. \textit{American Economic Review}, 109(6), 1991--2014.

\bibitem{Pindyck}
Pindyck, R.S. (2019). The social cost of carbon revisited. \textit{Journal of Environmental Economics and Management}, 94, 140--160.

\bibitem{Rennert}
Rennert, K., Errickson, F., Prest, B.C. et al. (2022). Comprehensive evidence implies a higher social cost of CO2. \textit{Nature}, 610, 687--692. 

%
%
%

\bibitem{Ricke2018} Ricke, K., Drouet, L., Caldeira, K.and Tavoni, M. (2018). Country-level social cost of carbon.
\textit{Nature Climate Change}, 8(10), 895--900.

%
%
%

\bibitem{SP02} Shrestha, R.K. and J.P. Feehan (2002/2003). Contributions to International Public Goods and the Notion of Country Size.
\textit{Public Finance Analysis}, 59(4), 551--559.

%

\bibitem{Worldbank} World Bank (2022). World Development Indicators. Retrieved from \url{https://databank.worldbank.org/source/world-development-indicators}.

\bibitem{WHO}World Health Organization (2022). Global Health Observatory data repository. Retrieved from \url{https://www.who.int/data/gho}.



\end{thebibliography}
\end{document}